\documentclass[pra,aps,twocolumn,superscriptaddress,english]{revtex4-1}
\usepackage[T1]{fontenc}
\usepackage[utf8]{inputenc}
\usepackage[british]{babel}
\usepackage{mathtools}
\usepackage{amsthm}
\usepackage{subcaption}
\usepackage{enumitem}

\usepackage{amsmath,amssymb,mathrsfs,hyperref}

\usepackage{graphicx}

\usepackage[matrix,frame,arrow]{xy}
\usepackage{amsmath}
\newcommand{\bra}[1]{\left\langle{#1}\right\vert}
\newcommand{\ket}[1]{\left\vert{#1}\right\rangle}

\newcommand{\transf}[1]{\ensuremath{\mathscr{#1}}}
\newcommand{\tA}{\transf A}

\newcommand{\tC}{\transf C}
\newcommand{\tD}{\transf D}
\newcommand{\tE}{\transf E}

\newcommand{\tM}{\transf M}

\newcommand{\tU}{\transf U}
\newcommand{\tV}{\transf V}

\newcommand{\tI}{\transf I}
\newcommand{\tP}{\transf P}
\newcommand{\sys}[1]{\ensuremath{\mathrm{#1}}}
\newcommand{\rA}{\sys A}
\newcommand{\rB}{\sys B}
\newcommand{\rC}{\sys C}

\newcommand{\rF}{\sys F}

\newcommand{\rL}{\sys L}
\newcommand{\rM}{\sys M}
\newcommand{\rN}{\sys N}
\newcommand{\rK}{\sys K}

\newcommand{\rX}{\sys X}

\newcommand{\st}{\mathsf{St}}
\newcommand{\purst}{\mathsf{PurSt}}

\newcommand{\eff}{\mathsf{Eff}}
\newcommand{\tr}{\mathsf{Tr}}
\newcommand{\refset}{\mathsf{Ref}}

\newtheorem{theorem}{Theorem}[section]
\newtheorem{lemma}{Lemma}[section]
\newtheorem{proposition}{Proposition}[section]
\newtheorem{corollary}{Corollary}[section]
\newtheorem{definition}{Definition}[section]
\DeclarePairedDelimiter{\norma}{\lVert}{\rVert}
\DeclarePairedDelimiter{\dket}{|}{\rangle\!\rangle}
\DeclarePairedDelimiter{\dbra}{\langle\!\langle}{|}
\DeclarePairedDelimiter{\ceil}{\lceil}{\rceil}
\def\dsc#1#2{\bb{#1}|{#2}\kk}
\newcommand{\cH}{\mathcal{H}}
\newcommand{\cF}{\mathcal{F}}
\newcommand{\sT}{\mathsf{T}}

\newcommand{\QTyp}{\mathsf{F}}
\newcommand{\Typ}{\mathsf{T}}
\newcommand{\Rng}{\mathsf{Rng}}
\newcommand{\vi}{\mathrm {\mathbf i}}
\newcommand{\Tr}{\mathrm{Tr}}
\newcommand{\Span}{\mathsf{Span}}
\newcommand{\Supp}{\mathsf{Supp}}
\newcommand{\proj}[1]{|{#1}\rangle\langle{#1}|}
\def\>{\rangle}

\def\rank#1{\operatorname{rank}(#1)}
\def\<{\langle}

\def\Tr{\operatorname{Tr}}
\def\kk{\rangle\!\rangle}
\def\bb{\langle\!\langle}
\def\Lin#1{\mathcal{L}(#1)}

\usepackage{color}

\begin{document}
\title{Shannon theory for quantum systems and beyond: information compression for fermions}

\author{Paolo \surname{Perinotti}}

\email{paolo.perinotti@unipv.it}

\affiliation{QUIT group, Physics Dept., Pavia University, and INFN Sezione di Pavia, via Bassi 6, 27100 Pavia, Italy}

\author{Alessandro \surname{Tosini}}

\email{alessandro.tosini@unipv.it}

\affiliation{QUIT group, Physics Dept., Pavia University, and INFN Sezione di Pavia, via Bassi 6, 27100 Pavia, Italy}

\author{Leonardo \surname{Vaglini}}

\email{leonardo.vaglini01@universitadipavia.it}

\affiliation{QUIT group, Physics Dept., Pavia University, and INFN Sezione di Pavia, via Bassi 6, 27100 Pavia, Italy}

\date{\today}
\begin{abstract}
  We address the task of compression of fermionic quantum
  information. Due to the parity superselection rule, differently from the case of 
  encoding of quantum information in qubit states, part of the
  information carried by fermionic systems is encoded in their delocalised
  correlations. As a consequence, reliability of a compression protocol must be 
  assessed in a way that necessarily accounts also for the preservation of correlations. 
  This implies that input/output fidelity is not a satisfactory figure of merit for fermionic
  compression schemes. We then discuss various aspects regarding the assessment of reliability
  of an encoding scheme, and show that entanglement fidelity in the fermionic case is capable of 
  evaluating the preservation of correlations, thus revealing itself
  strictly stronger than input/output fidelity, unlike the qubit case. We then introduce a fermionic
  version of the source coding theorem showing that, as in
  the quantum case, the von Neumann entropy is the minimal rate for
  which a fermionic compression scheme exists, that is reliable 
  according to the entanglement fidelity criterion.
\end{abstract}
\maketitle

\section{Introduction}
The task of \emph{data compression} addresses the primary question in
information theory as to how redundant is the information contained in a message
and to what extent the message can then be compressed.

In classical information theory this question is answered by the
source coding theorem \cite{6773024}, which establishes the
fundamental role of Shannon entropy in information theory and its
operational interpretation. The coding theorem recognizes the Shannon
entropy as the fundamental limit for the compression rate in the
i.i.d. setting. This means that if one compresses at a rate above the
Shannon entropy, then the compressed data can be recovered perfectly
in the asymptotic limit of infinitely long messages, while this is not
possible for compression rate below the Shannon entropy. As a result
the Shannon entropy, which can be intuitively thought of as the
uncertainty about the outcome of an experiment that we are going to
perform on a classical system, quantifies in a stringent way the
amount of ``non-redundant'' information that is encoded in the state
of the classical system, what one would definitely call
\emph{information content}.

In quantum information theory the Shannon entropy is replaced by the 
von Neumann entropy. In particular, the quantum source
coding theorem \cite{PhysRevA.51.2738} identifies von Neumann
entropy as the rate at which quantum compression can be reliably
achieved.  Consider a quantum information source described by a system
$\rA$ and density operator $\rho\in\st{(\rA)}$, with $\st{(\rA)}$ the
set of states of system $\rA$. The density operator describes the preparation of a state
$\sigma_i$ from any ensemble $\{p_i\sigma_i\}$, with probabilities $p_i$, such that 
$\sum_ip_i\sigma_i=\rho$. A quantum message of $N$ letters can
now be understood in terms of $N$ uses of the quantum source, that output a
sequence of $N$ states $p_{i_j}\sigma_{i_j}$, with $1\leq j\leq N$,
drawn independently. One instance of this preparation protocol thus produces $\sigma_{\vi}\coloneqq\bigotimes_{j=1}^N\sigma_{i_j}$, with probability $p_\vi\coloneqq\prod_jp_{i_j}$.
Each of the $N$ systems has density operator
$\rho$, and the density operator of the entire message is then given by
$\rho^{\otimes N}$. A compression scheme for messages from the above described source
consists of two steps. \emph{Encoding:} Alice encodes the system $\rA^{\otimes N}$
according to a compression map given by a channel
$\tE:\st(\rA^{\otimes N})\rightarrow \st(B)$, where $\rB$ is generally
a system with dimension $d_\rB(N)$ smaller then $\rA^{\otimes N}$. The
compression rate is defined as the asymptotic quantity
$R=\lim_{N\rightarrow\infty}\log_2d_{\rB}(N)/N$. Typically, one
estimates the ``size'' of the compressed message through the capacity of 
system $\rB$ given in terms of $\log_2d_\rB(N)$,
namely the number of qubits that are needed to simulate $\rB$. 
Alice then sends the system $\rB$ to Bob using $NR$
noiseless qubit channels.  \emph{Decoding:} Finally, Bob applies
a decompression map $\tD:\st(B)\rightarrow \st(\rA^{\otimes N})$ 
to the message encoded in system $\rB$, with the purpose of recovering 
the original message as reliably as possible.

As one might expect, the above scheme generally introduces an error in the
decoding: we now discuss the figure of merit by which we estimate the error introduced
by the compression scheme. In order to understand the operational meaning of the figure
of merit, think of a referee (Charlie) who prepares the states $\sigma_\vi$ with probability 
$p_\vi$, and receives the final states $\tD\tE(\sigma_\vi)$. The figure of merit that
we use corresponds to the probability that, after receiving Bob's final state, 
Charlie is able to distinguish it from the input one. 
For a single instance, this probability is a linear function of the trace-norm
distance $\|\sigma_\vi-\tD\tE(\sigma_\vi)\|_1$. The probability of successful discrimination 
is thus evaluated to $\sum_\vi p_\vi\|\sigma_\vi-\tD\tE(\sigma_\vi)\|_1=\|\rho^{\otimes N}-\tD\tE(\rho^{\otimes N})\|_1$. The protocol has then error $\epsilon$ if the 
compressed and decompressed states $\tD\tE(\sigma_\vi)$ are $\epsilon$-close
to the original states $\sigma_\vi$, in the trace-norm distance. In the case of qubits the above quantity 
is equivalent to fidelity, thanks to the Fuchs-van der Graaf inequalities~\cite{761271}.
The optimal quantum encoding will then make the error arbitrarily small for $N$ large enough, with
rate $R$ as small as possible. Schumacher's quantum source coding theorem shows that 
the optimal rate is equal to the von Neumann entropy $S(\rho)$ of the state $\rho$.

Another way to evaluate the error for a compression scheme is the following: 
Charlie prepares a purification of the density operator $\rho^{\otimes N}$
and sends the $N$ copies of system $\rA$ to Alice. Alice then sends
her share of the pure state to Bob, sending as
few qubits to Bob as possible. After decompressing the received qubits, Bob
shares an entangled state with Charlie. The quality of the compression scheme can then
be evaluated considering how well Charlie can distinguish the initial state from the final one, after
receiving Bob's $N$ systems. The probability that Charlie detects a compression error can be 
evaluated through the input/output fidelity. Again, Schumacher's theorem states that 
Alice can transfer her share of the pure state to Bob by sending $NS(\rho)$
qubits and achieving arbitrarily good fidelity, increasing the length
$N$ of the message. This second perspective answers the question
whether the compression protocol preserves the correlations that
system $\rA^{\otimes N}$ has with a remote system $\rC$. 

Equivalence of the two approaches in assessing the quality of a compression scheme 
shows that the ability to send quantum superpositions is equivalent to the
ability to send entanglement. In other terms, the amount of quantum
information preserved by a compression scheme represents the dimension of the
largest Hilbert space whose superpositions can be reliably
compressed, or equivalently the amount of entanglement that can be
reliably compressed.

According to the above discussion a crucial point in the compression
protocol is to quantify the reliability of the compression map
$\tC:=\tD\tE$, which in the asymptotic limit of $N\to\infty$ must
coincide with the identity map. In quantum theory checking the
reliability of the compression map looking only at its local action,
namely via the fidelity between states $\tC(\rho^{\otimes N})$ and
$\rho^{\otimes N}$, or at the effects on correlations, namely via
entanglement fidelity, is equivalent. This follows from \emph{local
  process tomography} of quantum theory where, given a map $\tC$ on
system $\rA$ one has
\begin{equation}\label{eq:local-process-tomography}
\begin{aligned}
(\tC\otimes\tI_{\rC})(\Psi) = & \Psi \qquad \forall \Psi \in \st(\rA\rC) \\
&\Leftrightarrow\\
\tC(\rho)=&\rho \qquad \forall \rho\in\st{(\rA)}.
\end{aligned}
\end{equation}
This equivalence is due to \emph{local 
discriminability}~\cite{PhysRevA.81.062348,PhysRevA.84.012311,DAriano:2017up} 
of quantum theory, where the discrimination of bipartite quantum states can always be
performed using local measurements only (this property is equivalent to the one
known in the literature as local tomography, or tomographic 
locality~\cite{Araki:1980tr,dakic_brukner_2011,Masanes_2011,Barnum:2014vt}). 
However, in the absence of local discriminability, 
a map preserving local states still can affect correlations with 
remote systems~\cite{DAriano2020information}. This
raises a crucial issue if one aims at studying the compression task beyond
quantum theory, where the reliability of a protocol generally needs to
be verified on extended systems. Indeed, in general, testing a compression 
scheme using ancillary systems is strictly stronger than testing them with local 
schemes.

As a first step in the direction of generalizing the compression
protocol to an arbitrary information theory, in this paper we consider
the case of fermionic systems as carriers of information. Fermionic
computation has been proposed in Ref.~\cite{Bravyi2002210} and later
studied in several works
\cite{Wolf2006,Banuls2007,Friis2013,DAriano2014a,PhysRevA.101.052326}. Differently
from quantum systems, fermions obey the \emph{parity superselection
  rule}. As a consequence, fermionic information theory does not
satisfy local discriminability, thus providing a physically relevant example 
of a theory where the task of compression is not straightforward. 
Indeed, in the case of study, a map $\tC$ that acts as the identity on local states
$\rho^{\otimes N}$ could still destroy the correlations with remote
systems, and then be mistakenly considered as a reliable compression map.

After reviewing the structure of fermionic quantum
information, we prove that the entanglement fidelity is a valid
criterion for the reliability of a fermionic compression map. We
then show an analogous of the quantum source coding theorem in the
fermionic scenario, showing that the minimal compression rate for
which a reliable compression scheme exists is the von Neumann entropy
of the fermionic state. We conclude therefore that the von Neumann
entropy provides the informational content of the state also in the
case of fermionic theory, namely in the presence of parity
superselection. The above result, however, is not a straightforward consequence
of Schumacher's source coding theorem.

\section{Fermionic information theory}

We now briefly review fermionic information theory. The systems of
the theory are made by local fermionic modes (LFMs). A
LFM is the counterpart of the qubit in quantum theory, and can be thought
of as a cell that can be either empty or occupied by a fermionic excitation. 
An $L$-LFMs system, denoted $\rL_\rF$, is described by $L$ fermionic fields
$\varphi_{i}$, satisfying the canonical anticommutation rule (CAR)
$\{\varphi_{i},\varphi_{j}^{\dagger} \} = \delta_{ij} I$,
$\{\varphi_{i},\varphi_{j}\}=0$ where $i,j=1,\dots,L$.  With these
fields one constructs the occupation number operators
$\varphi_{i}^{\dagger}\varphi_{i}$, which can be easily proved to 
have only eigenvalues 0 and 1. The common eigenvector $\ket\Omega$ of the operators
$\varphi_{i}^{\dagger}\varphi_{i}$, $i=1,\ldots,L$ with eigenvalue 0
defines the vacuum state $\ket{\Omega}\bra{\Omega}$ of $\rL_\rF$, 
representing the state in which all the modes are not excited. The fermionic vacuum state in terms of 
the field operators is given by $\ket{\Omega}\bra{\Omega}=\prod_{i=1}^{L}\varphi_{i}\varphi_{i}^{\dagger}$. By applying the operators $\varphi_{i}^{\dagger}$ to $\ket{\Omega}$ the corresponding $i$-th
mode is excited and, by raising $\ket{\Omega}$ in all possible ways, we
get the $2^L$ orthonormal vectors forming the Fock basis in the
occupation number representation: a generic element of this basis is
\begin{align}
\ket{n_{1},\dots,n_{L}}:=(\varphi_{1}^{\dagger}) ^{n_{1}}\dots
(\varphi_{L}^{\dagger}) ^{n_{L}}\ket{\Omega},
\label{eq:states}
\end{align}
with $n_i=\{0,1\}$ corresponding to the occupation number at the
$i$-th site. The linear span of these vectors corresponds to the
antisymmetric Fock space $\cF_{L}$ of dimension $d_{\cF_L}=2^L$.
Notice that the Fock space $\cF_L$ is isomorphic to the Hilbert space
$\mathcal{H}_L$ of $L$ qubits, by the trivial identification of the
occupation number basis with the qubit computational basis. This 
correspondence lies at the basis of the Jordan-Wigner
isomorphism~\cite{Jordan1928,Verstraete2005,Pineda2010} typically used
in the literature to map fermionic systems to qubits systems and
vice-versa. We recall here the definition of the Jordan-Wigner map
\begin{equation}
\begin{aligned}
&J_L(\varphi_i)=\left(\bigotimes_{l=1}^{i-1}\sigma^z_{l}\right)\otimes\sigma^-_{i}\otimes\left(\bigotimes_{k=i+1}^{L}I_k\right),\\
&J_L(\varphi^\dag_i)=J_L(\varphi_i)^\dag,\\
&J_L(XY)=J_L(X)J_L(Y),\\ 
&J_L(aX+bY)=aJ_L(X)+bJ_L(Y),
\end{aligned}
\end{equation}
with $X,Y$ linear combinations of products of field operators on the $L$-LFMs, and where 
we used the standard notation for Pauli sigma operators. In the following we will drop the dependence 
on the number of LFMs in the Jordan-Wigner map, namely we will  write $J(X)$ in place of $J_L(X)$, 
when it will be clear from the context. Notice that the Jordan-Wigner 
isomorphism is implicitly defined in Eq.~\eqref{eq:states}, and, as such, it depends on the 
arbitrary ordering of the modes. All such representations are unitarily equivalent.

Differently from standard qubits, fermionic systems satisfy the parity 
superselection rule~\cite{Schuch2004,Kitaev2004,Schuch2004a,DAriano2014a,fermionic_theory}. 
One can decompose the Fock space $\cF_L$ of system $\rL_\rF$ in the direct sum
$\cF_L=\cF_L^e\oplus\cF_L^o$, with $\cF^e_L$ and $\cF^o_L$ the spaces
generated by vectors with even and odd total occupation number,
respectively. The convex set of states $\st{(\rL_{\rF})}$ corresponds,
in the Jordan-Wigner representation,
to the set of density matrices  on $\cF_{L}$ of the form
$\rho=\rho_e+\rho_o$, with $\rho_e,\rho_o\geq 0$,
$\Tr[\rho_o]+\Tr[\rho_e]\leq 1$ and with $\rho_e$ and $\rho_o$ having
support on $\cF_L^e$ and $\cF_L^o$, respectively, and pure states
are represented by rank one density operators. 
Moreover, the density matrices representing the states 
represent linear combinations of
products of an even number of field operators (see appendix \ref{app:appstates} and \cite{fermionic_theory} for further details). Viceversa, every linear 
combination of products of an even number of field operators that is represented by a density matrix is an admissible state.
Analogously, effects in the set $\eff{(\rL_{\rF})}$ are represented by positive operators on
$\rL_{\rF}$ of the form $a=a_e+a_o$, with $a_e$ and $a_o$ having
support on $\cF_L^e$ and $\cF_L^o$, respectively.  Notice that set of
states and effects of system $\rL_\rF$ have dimension
$d^2_{\cF_L}/2=2^{2L-1}$, compared to the quantum case where the set
of states and effects associated to the Hilbert space $\mathcal{H}_L$
of $L$ qubits has dimension $d^2_{\mathcal{H}_L}=2^{2L}$. 

Given a state $\rho\in\st(\rL_\rF)$ we define the refinement set of
$\rho$ as
$\refset(\rho):=\{\sigma\in\st(\rL_\rF)|\exists
\tau\in\st(\rL_\rF):\rho=\sigma+\tau\}$,
and a state is pure when all the elements in the refinement are
proportional to the state itself. In the following we will denote by
$\purst(\rL_\rF)$ and $\st_{1}(\rL_\rF)$ the set of pure states and
the set of normalized states (of trace one) of system $\rL_\rF$,
respectively. 

Given two fermionic systems $L_{\rF}$ and $M_{\rF}$, we introduce the composition
of the two as the system made of $K\equiv L+M$ LFMs, denoted with the symbol
$\rK_\rF\coloneqq\rL_{\rF}\boxtimes\rM_{\rF}$, or simply $\rK_\rF\coloneqq\rL_{\rF}\rM_{\rF}$.
We use the symbol $\boxtimes$  to distinguish the fermionic parallel composition rule from the quantum one, corresponding to the tensor product $\otimes$.

Given a state $\Psi \in \st(\rL_{\rF}\rM_{\rF})$, one can discard the subsystem
$\rM_{\rF}$ and consider the marginal state, which we denote by $\sigma:=\Tr^{f}_{\rM_{\rF}}(\Psi)$.
We use the symbol $\Tr^{f}_{\rM_{\rF}}$ to denote the fermionic partial trace on the subsystem $\rM_{\rF}$. This is computed by performing the following steps (see ref. \cite{fermionic_theory} for further details): (i) drop all those terms in $\Psi$ containing
an odd number of field operators in any of the LFMs in $\rM_{\rF}$; (ii) remove all the field operators
corresponding to the LFMs in $\rM_{\rF}$ from the remaining terms.
The fermionic trace $\Tr^f(\rho)$ of a state
$\rho\in\st(\rM_{\rF})$ is then defined as
a special case of the partial one, corresponding to the case in which $L=0$.

Finally, the set of transformations from $\rL_\rF$ to $\rM_\rF$,
denoted by $\tr(\rL_\rF\rightarrow \rM_\rF)$, is given by completely
positive maps from $\st(\rL_\rF)$ to $\st(\rM_\rF)$ in the Jordan-Wigner representation. Moreover, we
denote by $\tr_1(\rL_\rF\rightarrow \rM_\rF)$ the set of deterministic
transformations, also called \emph{channels}, from $\rL_\rF$ to $\rM_\rF$,
corresponding to trace-preserving completely positive maps. Like in quantum theory, any fermionic transformation
$\tC\in \tr(\rL_\rF\rightarrow \rM_\rF)$ can be expressed in Kraus
form $\tC(\rho)=\sum_i C_i\rho C^\dag_i$, with deterministic
transformations having Kraus operators $\{C_i\}$ such that
$J(\sum_{i}C_{i}^{\dagger}C_{i})=I_{\cH_L}$, $I_{\cH_L}$ denoting the
identity operator on $\cH_{L}$. For a map $\tC\in\tr(\rL_\rF\rightarrow \rM_\rF)$ with Kraus operators 
$\{C_i\}$, we define its Jordan-Wigner representative $J(\tC)$ as the quantum map with Kraus operators 
$\{J(C_i)\}$. Now, given two transformations $\tC\in\tr(\rL_\rF\rightarrow \rM_\rF)$ and
$\tD\in\tr(\rK_\rF\rightarrow \rN_\rF)$, we denote by $\tC\boxtimes\tD\in\tr(\rL_\rF\rK_\rF\to\rM_\rF\rN_\rF)$ 
the \emph{parallel composition} of $\tC$ and $\tD$, with Kraus operators $\{C_iD_j\}$, where $\{C_i\}$ are Kraus 
operators for $\tC$ and $\{D_j\}$ for $\tD$. We observe that in the 
Jordan-Wigner representation one generally has $J_{L+K}(C_iD_j)\neq J_L(C_i)\otimes J_K(D_j)$, and 
$J_{L+K}(\tC\boxtimes\tD)\neq J_L(\tC)\otimes J_K(\tD)$. If $\tC$ is a transformation 
in $\tr(\rL_\rF\rightarrow \rM_\rF)$, its extension to a composite system $\rL_\rF\rN_\rF$, is given by 
$\tC\boxtimes\tI$, with $\tI$ the identity map of system $\rN_\rF$---whose Jordan-Wigner representative is the 
quantum identity map---and its Kraus operators involve field operators on the $\rL_{\rF}$ modes only. It is worth 
noticing that, despite the Jordan-Wigner representative of 
this map is not necessarily of the form $J_L(\tC)\otimes \tI$, upon suitable choice of the ordering of the 
LFMs that defines the representation, one can always reduce to the case where, actually, 
$J_{L+N}(\tC\boxtimes\tI)=J_L(\tC)\otimes \tI$.

As a special case of the above composition rule, one can define 
$\rho\boxtimes\sigma\coloneqq\rho\sigma\in\st(\rL_\rF\rM_\rF)$ for
the parallel composition of states $\rho\in\st(\rL_\rF)$ and $\sigma\in\st(\rM_\rF)$, and similarly
$a\boxtimes b\coloneqq ab\in\eff(\rL_\rF\rM_\rF)$ for the parallel composition of effects 
$a\in\eff(\rL_\rF)$ and $b\in\eff(\rM_\rF)$.

A useful characterization of fermionic maps in $\tr(\rL_\rF\rightarrow \rL_\rF)$, proved in Ref.~\cite{DAriano2014}, is the following:
\begin{proposition}[Fermionic transformations]\label{prop:fermionickraus}
  All the transformations in $\tr(\rL_\rF\rightarrow \rL_\rF)$ with Kraus operators being linear combinations of
  products of either an even number or an odd number of field
  operators are admissible fermionic transformations. Viceversa, each
  admissible fermionic transformation in $\tr(\rL_\rF\rightarrow \rL_\rF)$ has Kraus operators being superpositions of
  products of either an even number or an odd number of field
  operators.
\end{proposition}

\begin{corollary}[Fermionic effects]\label{cor:fermioniceff}
Fermionic effects are positive operators bounded by the identity operator that are linear combinations of 
products of an even number of field operators. Viceversa, every linear combination of 
products of an even number of field operators that is represented by a positive operator bounded by the identity is a fermionic effect.
\end{corollary}

The corollary follows immediately from Proposition~\ref{prop:fermionickraus}, since an effect $A$ 
is obtained as a fermionic transformation $\tA$ followed by the discard map, i.e.~the trace. Thus
\begin{align*}
\Tr[\rho A]=\Tr[\tA(\rho)]=\sum_i\Tr[K_i\rho K_i^\dag]=\Tr[\rho\sum_iK_i^\dag K_i],
\end{align*}
namely $A=\sum_iK_i^\dag K_i$. Having the polynomial $K_i$ a definite parity 
(though not necessarily the same for every $i$), $A$ is an even polynomial.

In the following we denote by
$\Lin{\cH_L}$ the set of linear operators on the Hilbert space $\cH_L$ of $L$-qubits
and by $\Lin{\cH_L, \cH_M}$ the set of linear operators from $\cH_L$
to $\cH_M$. It is useful to introduce the isomorphism
between operators $X$ in $\Lin{\cH_L,\cH_M}$ and vectors $|X \kk$ in $\cH_M \otimes \cH_L$ given by
\begin{equation}\label{eq:isom}
  |X \kk = (X \otimes I_{\cH_L}) | I_{\cH_L} \kk = (I_{\cH_M} \otimes
  X^T)  |I_{\cH_M} \kk,
\end{equation}
where $I_{\cH_L}$ is the identity operator in $\cH_L$,
$|I_{\cH_L}\kk \in \cH_L^{\otimes 2}$ is the maximally entangled
vector $|I_{\cH_L}\kk = \sum_{l} |l\>|l\>$ (with $\{|l\>\}$ a fixed
orthonormal basis for $\cH_L$), and $X^T \in \Lin{\cH_M, \cH_L}$ is
the transpose of $X$ with respect to the two fixed bases chosen in
$\cH_L$ and $\cH_M$. Notice also the useful identity
\begin{align}
  Y\otimes Z |X\kk=|YXZ^T\kk,
\end{align}
where $X\in\Lin{\cH_L,\cH_M}$, $Y\in\Lin{\cH_M,\cH_N}$ and
$Z\in\Lin{\cH_L,\cH_K}$. Moreover, for $X,Y\in\Lin{\cH_L,\cH_M}$, one
has $\Tr_{\cH_L}[|X\kk\bb Y|]=XY^\dag$, and
$\Tr_{\cH_M}[|X\kk\bb Y|]=X^TY^*$. We remark that, in the above paragraph, we are dealing with
abstract linear operators on an Hilbert space, disregarding their possible interpretation as Jordan-Wigner 
representatives of some fermionic operator.

A notion that will be used in the following is that of states
dilation.
\begin{definition}[Dilation set of a state $\rho$]\label{def:purification} For any
  $\rho\in\st{(\rL_\rF)}$, we say that
  $\Psi_\rho\in\st(\rL_\rF\rM_\rF)$ for some system $\rM_\rF$,
  is a dilation of $\rho$ if $\rho=\Tr^{f}_{\rM_{\rF}}[\Psi_\rho]$. 
We denote by $D_\rho$ the set of all possible dilations of $\rho$. A pure dilation 
$\Psi_\rho\in\purst(\rL_\rF\rM_\rF)$ of $\rho$
is called a purification.
\end{definition}
Naturally, any purification of $\rho$ belongs to $D_\rho$, more precisely
the set of purifications of $\rho$ is the subset of $D_\rho$
containing pure states.
A main feature of quantum theory that is valid also for fermionic
systems is the existence of a purification of any state, 
that is unique modulo channels on the purifying system.
\begin{proposition}[Purification of states]\label{def:purification} For every
  $\rho\in\st{(\rL_\rF)}$, there exists a purification
  $\Psi_\rho\in\purst(\rL_\rF\rM_\rF)$ of $\rho$
  for some system $\rM_\rF$.
  Moreover, the purification is
  unique up to channels on the purifying system: if
  $\Psi_\rho\in\purst(\rL_\rF\rM_\rF)$ and
  $\Phi_\rho\in\purst(\rL_\rF\rK_\rF)$ are two purifications of $\rho$
  then there exists a channel $\tV\in\tr_1(\rM_\rF\rightarrow\rK_\rF)$
  such that $(\tI_{\rL_\rF}\boxtimes \tV)(\Psi_\rho)=\Phi_\rho$.
\end{proposition}
\begin{proof}
It can be easily verified that every purification of
$\rho\in\st{(\rL_\rF)}$, having even part $\rho_e$ and odd part
$\rho_o$, can be obtained in terms of the minimal one
$J^{-1}(|F\kk\bb
F|)\in\purst(\rL_\rF\rM_\rF)$,
with $F=J(\rho)^{\frac{1}{2}}$, $M=\ceil{\log_2{2r}}$ and $r=\max(\rank{\rho_e}, \rank{\rho_o})$.
Now, let $\Psi_\rho\in\purst(\rL_\rF\rM_\rF)$ and $\Phi_\rho\in\purst(\rL_\rF\rK_\rF)$ 
be two purifications of $\rho$. If $M=K$, let us choose the
ordering defining the Jordan-Wigner isomorphism of Eq.~\eqref{eq:states} in such a way that 
the modes in the purifying systems $\rM_\rF$ precede the modes of $\rL_{\rF}$. Then, using the
quantum purification theorem, we know that there exists a reversible map $\tU$ with unitary
Kraus $U$ such that
$\dket{F_{\rho}}=(U\otimes I)\dket{P_{\rho}}$, where
\begin{align*}
\dket{F_{\rho}}\dbra{F_{\rho}}=J(\Phi_\rho),\quad \dket{P_{\rho}}\dbra{P_{\rho}}=J(\Psi_\rho). 
\end{align*}
The unitary $U$ can be chosen in such a way that  
$J^{-1}(\tU)$ is an admissible fermionic map, namely in such a way that it respects the parity superselection 
rule (see Lemma~\ref{lem:techlemma0} in Appendix~\ref{app:app1}). Moreover, due to Lemma~\ref{lem:techlemma} in Appendix~\ref{app:app1},  $J^{-1}(U\otimes I)$ cannot contain field operators on the 
 modes in $\rL_{\rF}$, and is then local on the purifying system $\rK_{\rF}$.
Now, let $K>M$. Then, we can consider a pure state $\omega$ on the $K-M$ modes and take
the parallel composition $\Psi_{\rho}\boxtimes\omega$. This is still a purification of $\rho$,
and by the previous argument, there exists a reversible channel $\tU\in\tr_{1}(\rK_\rF\rightarrow \rK_\rF)$
such that $\Phi_{\rho}=(\tI_{\rL_{\rF}}\boxtimes\tU) (\Psi_{\rho}\boxtimes\omega)=(\tI_{\rL_{\rF}}
\boxtimes\tV)(\Psi_{\rho})$ where $\tV$ is the channel defined by $\tV=\tU(\tI\boxtimes\omega)$.
If $K<M$, we consider $\Phi_{\rho}\boxtimes\omega$, where $\omega$ is any pure state on $N=M-K$ modes system,
and we have $\Phi_{\rho}\boxtimes\omega=(\tI_{\rL_{\rF}}\boxtimes\tU)(\Psi_{\rho})$. Now we discard
the additional modes, and the channel connecting the purifications is the sequential composition of $\tU$ and the discarding 
map: $\tV:=(\tI_{\rK_\rF}\boxtimes\Tr^{f}_{\rN_{\rF}})\tU$. 
\end{proof} 

The main difference between fermionic and quantum information lies in the
notion of what Kraus operators correspond to local maps. While in the case
of qubit systems local maps acting on the $i$-th qubit of a composite system 
have Kraus operators that can be factorized as a non trivial operator on the $i$-th 
tensor factor $\mathbb C^2$ of the total Hilbert space, in the case of the fermionic 
Fock space $\cF_L$ a local transformation on the $i$-th mode can be represented in the Jordan-Wigner 
isomorphism by operators that act non trivially on factors $\mathbb C^2$ different from the $i$-th one.
This fact is the source of all the differences between the theory of qubits 
and fermionic theory, including superselection and features that it affects, 
such as the notion of entanglement~\cite{DAriano2014} and local states discrimination
protocols~\cite{fermLOCC1,fermLOCC2}. Due to parity superselection, fermionic theory does not
satisfy \emph{local process tomography}, namely the property stating 
that two transformations $\tC_1,\tC_2\in\tr(\rL_\rF\rightarrow\rM_\rF)$ 
are equal iff they act in the same way on the local states in $\st(\rL_\rF)$, 
namely $\tC_1(\rho)=\tC_2(\rho)$ for every $\rho\in\st(\rL_\rF)$ (see for
example Eq.~\eqref{eq:local-process-tomography} in the introduction on
the equality between the compression map $\tC$ and the identity map).
As a consequence, fermionic theory also violates \emph{local tomography}.
A typical example of a transformation that is locally equivalent to the
identity but differs from it when extended to multipartite systems is the parity transformation,
as shown in the following. Let us consider a single fermionic mode system
$\mathrm{1}_\rF$, whose possible states are constrained to be of the
form $J(\rho)=q_{0}\ket{0}\bra{0}+q_{1}\ket{1}\bra{1}$ by the parity
superselection rule. Let $P_{0}$ and $P_{1}$ be the projectors on
$\ket{0}$ and $\ket{1}$ respectively, namely on the even and odd
sector of the Fock space. The parity transformation $\tP$, that in the Jordan-Wigner representation $J(\tP)$
has Kraus operators $P_{0}$ and $P_{1}$, acts as the
identity $\tI_{\mathrm{1}_\rF}$ when
applied to states in $\st{(\mathrm{1}_\rF)}$. However, taking the
system $\mathrm{2}_\rF$ and considering the extended transformation
$\tP\boxtimes\tI_{\mathrm{1}_\rF}$ on $\st{(\mathrm{2}_\rF)}$ one
notices that $\tP$ differs from the identity map
$\tI_{\mathrm{1}_\rF}$. Indeed, the state $J^{-1}(\ket{\Psi}\bra{\Psi})$, with
$\ket{\Psi}=\frac{1}{\sqrt{2}}(\ket{00}+\ket{11})$ is a legitimate
fermionic state in $\st{(\mathrm{2}_\rF)}$, and one can
straightforwardly verify that

\begin{align*}
(\tP\boxtimes\tI_{\mathrm{1}_\rF})[J^{-1}(\ket{\Psi}\bra{\Psi})]&=\frac{1}{2}J^{-1}(\ket{00}\bra{00}+\ket{11}\bra{11})
\\
&\neq J^{-1}(\ket{\Psi}\bra{\Psi}).
\end{align*}

\subsection{Identical channels upon-input of
  $\rho$}\label{sec:close-maps}

In the following we will be interested in quantitatively assessing how closely a channel (the
compression map) resembles another one (the identity map), provided that we know 
that the input state corresponds to a given $\rho$. To this end we introduce the notion of
identical fermionic channels upon-input of $\rho$.

Given two fermionic channels
$\tC_1,\tC_2\in\tr_1{(\rL_\rF\to\rM_\rF)}$ and a state
$\rho\in\st(\rL_\rF)$, we say that $\tC_1$ and $\tC_2$ are equal upon-input of $\rho$ if
\begin{equation}\label{eq:equal-fermionic-channels}
  (\tC_1\boxtimes \tI)(\Sigma)=(\tC_2\boxtimes \tI)(\Sigma)\qquad \forall\Sigma\in\refset(D_\rho).
\end{equation}
Operationally,
this means that one cannot discriminate between $\tC_1$ and $\tC_2$
when applied to any dilation $\Psi_\rho$ of the state $\rho$,
independently of how $\Psi_\rho$ has been prepared. Suppose that
$\Psi_\rho\in D_\rho$ was prepared as $\Psi_\rho=\sum_i\Sigma_i$, for
some refinement of $\Psi_\rho$. Even using the knowledge of the
preparation, one cannot distinguish between $\tC_1$ and
$\tC_2$. Notice that, differently from the quantum case 
here it is necessary to check the identity between channels on bipartite systems. 

Following the above definition, one can quantify how close two channels are.  One
has that $\tC_1$ and $\tC_2$ are $\varepsilon$-close upon-input of
$\rho$ if
\begin{align}\label{eq:close-fermionic-channels}
&\sum_i\norma{[(\tC_1-\tC_2)\boxtimes \tI](\Sigma_i)}_1\leq\varepsilon\quad\forall \{\Sigma_i\}:\ \sum_i\Sigma_i\in D_\rho,
\end{align}
where $\norma{X}_1$ is the $1$-norm of $J(X)$
%~\footnote{In the
%  paper $\norma{T}_{p}$ denotes the Shatten $p$-norm of the operator
%  $J(T)\in \Lin{\cF_{L},\cF_{M}}$ which is defined as:
%\begin{equation*}
%  \norma{T}_{p}={\rm Tr}[J(T^{\dagger}T)^{p/2}]^{1/p},
%\end{equation*}
%and the $\infty$-Shatten norm corresponds to the sup norm of $J(T)$ (see for instance \cite{bhatia97})
%\begin{equation*}
%\norma{T}_{\infty}=\sup \{ J(T)\phi : \phi \in \cF_L, \norma{\phi}\leq 1\}.
%\end{equation*}
%%A relevant inequality is~\cite{QIwatrous}
%%\begin{equation}
%%\norma{STR}_{p} \leq \norma{S}_{\infty} \norma{T}_{p} \norma{R}_{\infty},
%%\label{eq:normprop}
%%\end{equation}
%%where $S$, $T$, $R$ are linear operators defined on suitable Fock
%%spaces.
%}
. One can straightforwardly prove that the trace distance
$d(\rho,\sigma):=\frac{1}{2}\norma{\rho-\sigma}_1$ has a clear
operational interpretation in terms of the maximum success probability
of discrimination between the two states $\rho$ and
$\sigma$. Eq.~\eqref{eq:close-fermionic-channels} provides then an
upper bound for the probability of discriminating between $\tC_1$ and
$\tC_2$ when applied to the dilations of $\rho$, including their
refinements: $\tC_1$ and $\tC_2$ cannot be distinguished with a succes
probability bigger than
$\frac{1}{2}+\frac{1}{4}\varepsilon$. Accordingly, a sequence of
channels $\tC_N\in\tr_1{(\rL_\rF\to\rM_\rF)}$ converges to the channel
$\tC\in\tr_1{(\rL_\rF\to\rM_\rF)}$ upon-input of $\rho$ if
\begin{equation*}
  \lim_{N\to\infty}\norma{[(\tC_N-\tC)\boxtimes \tI](\Sigma)}_1=0\quad \forall \Sigma\in\refset(D_\rho).
\end{equation*}

\section{Fermionic compression}

Consider now a system $\rL_\rF$ and let
$\rho\in\st{(\rL_\rF)}$ be the generic state of the system. As usual
the source of fermionic information is supposed to emit $N$
independent copies of the state $\rho$ . A fermionic compression
scheme $(\tE_N,\tD_N)$ consists of the following two steps:

\begin{enumerate}[leftmargin=*]
\item Encoding: Alice encodes the system $\rL_\rF^{\boxtimes N}$ via 
  a channel $\tE_N:\st(\rL_\rF^{\boxtimes N})\rightarrow \st(\rM_\rF)$,
  where the target system is generally a system of $M$-LFMs. The map
  $\tE_N$ produces a fermionic state
  $\tE(\rho^{\boxtimes N})$ with support $\Supp(\tE(\rho^{\boxtimes N}))$
  on a Fock space $\cF_M$ of dimension $d_{\cF_M}(N)$ smaller than the
  one of the original state $\rho^{\boxtimes N}$. The compression rate
  is defined as the quantity
\begin{align*}
  R=\log_2 d_{\cF_M}(N)/N.
\end{align*}
Alice sends the system $\rM_\rF$ to Bob using $N\ceil{R}$ noiseless fermionic
channels.
\item Decoding: Finally Bob sends the system $\rM_\rF$ through a decompression
  channel $\tD_N:\st(\rM_\rF)\rightarrow \st(\rL_\rF^{\boxtimes N})$.
\end{enumerate}

% Let us start by giving a formal definition of compression scheme in
% FT, where the fermionic information source is given by a state
% $\rho\in\st(\rA)$.
% \begin{definition}
%   Let $\rA$ be an L-LFMs system and let $R>0$. A compression scheme
%   with rate $R$ is a pair of channels $(\tE^{N},\tD^{N})$,
%   $\tE^{N}:\rA^{\otimes N}\rightarrow \rB_{N}$ and
%   $\tD^{N}:\rB_{N}\rightarrow \rA^{\otimes N}$, where $\rB$ is an
%   $\ceil{R}$-LFMs system. $\tE^{N}$ and $\tD^{N}$ are called encoding
%   and decoding respectively.
% \end{definition}

The scheme $(\tE_N,\tD_N)$ overall transforms the $L^{\boxtimes N}$ LFMs, with a compression map
$\tC_N:=\tD_N\tE_N$. The latter can be more or less ``good'' (in a sense that
will be precisely defined) in preserving the information which is
contained in the system, depending on $\rho$ itself. The goal now is to define
the notion of reliable compression scheme once we are provided with an
information source $\rho$.

\subsection{Reliable compression scheme}

The aim of a compression scheme, besides reducing the amount of information 
carriers used, is to preserve all the information that is
possibly encoded in a given state $\rho$.  What we actually mean is not
only to preserve the input state and keep track of the correlations of
our system with an arbitrary ancilla, but also to preserve these
informations for any procedure by which the input system and its correlations have 
been prepared. In other words, even the agent that prepared the system along with possible 
ancillary systems, must have a small success probability in detecting the effects
of compression on the original preparation. This amounts to require that the compression channel 
$\tC_{N}$ must be approximately equal to the identity channel upon-input of
$\rho$, and more precisely that in the limit of $N\to\infty$ the two
channels must coincide upon-input of $\rho$.

In Section~\ref{sec:close-maps} we introduced the notion of
$\varepsilon$-close channels upon-input of $\rho$. This notion can now
be used to quantify the error, say $\varepsilon$, introduced by the map
$\tC_N$ in a compression protocol given the source $\rho$. According to
Eq.~\eqref{eq:close-fermionic-channels} we have indeed the following
definition of a reliable compression scheme 
\begin{definition}[Reliable compression scheme] Given a state $\rho\in\st(\rL_\rF)$, a
  compression scheme $(\tE_{N},\tD_{N})$ is $\varepsilon$-reliable if $\sum_i\left\lVert(\tC_{N}\boxtimes\tI)(\Sigma_i)-\Sigma_{i}\right\rVert<\varepsilon$
  for every $\{\Sigma_i\}$ such that $\sum_i\Sigma_i\in D_{\rho^{\boxtimes N}}$, 
  where $\tC_N:=\tD_{N}\tE_{N}$.
\end{definition}

It is clear from the definition that in order to check the reliability of a
fermionic compression map one should test it on states of an arbitrary
large system, since the dilation set $D_{\rho^{\boxtimes N}}$ includes
dilations on any possible ancillary system. It is then necessary to find a
simpler criterion to characterize the reliability of a compression
scheme. Let us start with a preliminary definition.
\begin{definition}\label{def:squareroot}
Let $\rho\in\st(\rL_\rF)$. We define its \emph{square root} $\rho^\frac12$ as follows
  \begin{align}
  \rho^\frac12\coloneqq J^{-1}[J(\rho)^\frac12].
  \end{align}
\end{definition}
One can easily prove that the square root of a fermionic state is well defined,
i.e.~it does not depend on the particular Jordan-Wigner representation $J$ chosen 
(see Appendix~\ref{app:well-defined-operations}).
In the following we show that a useful criterion for reliability can be
expressed via \emph{entanglement fidelity}:

\begin{definition}[Entanglement fidelity]\label{def:ent_fid}
  Let $\rho\in\st_1(\rL_\rF)$,
  $\tC\in\tr_{1}(\rL_\rF\rightarrow\rM_\rF)$ and
  $\Phi_{\rho}\in\purst(\rL_\rF\rK_\rF)$ be any purification of
  $\rho$. The entanglement fidelity is defined as and
  $F(\rho,\tC)=F(\Phi_{\rho},(\tC\boxtimes\tI)(\Phi_{\rho}))^{2}$, where
$F(\rho,\sigma):=\Tr[J(\rho^{1/2}\sigma\rho^{1/2})^{1/2}]$ denotes the
  Uhlmann's fidelity between states
  $\rho,\sigma\in\st_1(\rL_\rF)$.
\end{definition}

We notice that the Uhlmann fidelity of fermionic states is well defined, namely it is independent of the ordering of the fermionic modes (see Appendix~\ref{app:well-defined-operations}). As a consequence also the Entanglement fidelity, given in terms of the Uhlmann one, must be well defined.

Since by definition the Uhlmann fidelity of fermionic states coincides with the one of their Jordan-Wigner 
representatives and the same for their trace-norm distance, given $\rho,\sigma\in \st(\rL_\rF)$, the Fuchs-van 
der Graaf inequalities~\cite{761271} hold as a trivial consequence of their quantum counterparts
\begin{equation}\label{eq:fuchs-vandergraaf}
1-F(\rho,\sigma) \leq \frac{1}{2}\norma{\rho-\sigma}_{1} \leq \sqrt{1 - F(\rho,\sigma)^{2}}.
\end{equation}
  
The following proposition summarizes the main properties of fermionic entanglement fidelity that
will be used in the remainder.
\begin{proposition}
%[Properties of fermionic entanglement fidelity]
Let $\rho\in\st_1(\rL_\rF)$,
  $\tC\in\tr_{1}(\rL_\rF\rightarrow\rL_\rF)$ and
  $\Phi_{\rho}\in\purst(\rL_\rF\rK_\rF)$ be any purification of
  $\rho$. Entanglement fidelity has the following properties.
\begin{enumerate}[leftmargin=*]
\item $F(\rho,\tC)$ is independent of the particular choice for the
  purification $\Phi_{\rho}$.  
\item If the ordering is chosen in such a way that the $L$ modes are all before the purifying ones,
the following identity holds:
\begin{equation}
F(\rho,\tC)=\sum_{i}|\Tr[J(\rho) C_{i}]|^{2} \label{EntFididentity}
\end{equation}
for arbitrary Kraus decomposition
$J(\tC)=\sum_{i}C_{i}\cdot C_{i}^{\dagger}$ of the Jordan-Wigner representative $J(\tC)$.
From the second inequality in \eqref{eq:fuchs-vandergraaf} it follows that, if
$F(\rho,\tC) \geq 1 - \delta$, one has
\begin{equation}
\norma{(\tC\boxtimes\tI_{\rC})(\Phi_{\rho})-\Phi_{\rho}}_{1} \leq 2\sqrt{\delta}
\label{eq:rhoineq}
\end{equation}
for every purification $\Phi_{\rho}$ of $\rho$.
\end{enumerate}
\end{proposition}
\begin{proof}
%By definition, the entanglement fidelity can be written in the following way
%\[
%F(\rho,\tC)=\Tr[\Phi_{\rho}(\tC\boxtimes\tI)(\Phi_{\rho})].
%\]
%Since the reordering of the modes can be thought of as a unitary change of basis and the trace
%is basis independent, this is a well defined quantity \h{(see Appendix~\ref{app:well-defined-operations})}. 
Let $\Phi_{\rho}\in\purst(\rL_{\rF}\rM_{\rF})$ be a purification of $\rho$.
If we choose the trivial ordering for the LFMs, the Kraus of $J(\tC\boxtimes\tI)$ are of the form
$C_{i}\otimes I$. Moreover, since the minimal purification $\dket{F}\dbra{F}$ (introduced in the proof
of proposition \ref{def:purification}) and
$J(\Phi_{\rho})$ both purify the same quantum state, they are connected through an isometry $V$.
Recalling that for quantum states $\ket{\psi}\bra{\psi}$ and $\sigma$ the quantum Uhlmann fidelity is given by 
$F(\ket{\psi}\bra{\psi},\sigma)=\bra{\psi}\sigma\ket{\psi}^{1/2}$, we find
\[
\begin{aligned}
F(\rho,\tC)=&\sum_{i}\Tr(\dket{FV^{T}}\dsc{FV^{T}}{C_{i}FV^{T}}\dbra{C_{i}FV^{T}})\\
 =&\sum_{i}|\Tr(\dket{C_{i}FV^{T}}\dbra{FV^{T}})|^{2}= \\
 =& \sum_{i}|\Tr[J(\rho) C_{i}]|^{2},
\end{aligned}
\]
namely, the claimed formula in \eqref{EntFididentity}.
Since $\Phi_{\rho}$ is arbitrary, this also implies independence from the choice of the purification.
\end{proof}

% ----
% In quantum theory the entanglement fidelity is a good quantity in
% order to establish the reliability of a compression scheme, indeed it
% is proved that high entanglement fidelity $F(\rho,\tC)$ implies that
% the channel $\tC$ is close to the identity according to BLA . Here we
% prove an analogous result in the fermionic case

% \begin{theorem}
% Let $\rho\in\st_{1}(\rA)$ and let $\tE$ be a channel. $\forall\varepsilon>0$ there exists 
% $\delta>0$ such that $F(\rho,\tE)>1-\delta$ implies the following bound
% \begin{equation}
% \left\lVert \tE(\sigma)-\sigma\right\rVert_{1}\leq\varepsilon
% \end{equation}
% For any $\sigma\in\refset(\rho)$ 
% \end{theorem}

In quantum theory a compression scheme
$(\tE_{N},\tD_{N})$ is reliable when the entanglement fidelity
$F(\rho^{\boxtimes N},\tC_{N})$, with $\tC_N:=\tD_N \tE_{N}$,
approaches $1$ as $N\rightarrow \infty$. Here we prove an analogous
reliability criterion for the fermionic case. 

%To this end we show that
%if a fermionic channel has high entanglement fidelity on a given
%state, then it preserves all decompositions of any state in
%$D_{\rho}$. Before stating this result we prove a
%preliminary lemma.
%\begin{lemma}\label{lem:ref} 
%Let $\rho\in\st(\rL_\rF)$. Then
%  $\refset(D_\rho)=D_{\refset(\rho)}$
%\end{lemma}
%
%\begin{proof}
%  First notice that
%  \begin{enumerate}
%  \item $\Psi\in\st(\rL_\rF\rM_\rF)$ is in $D_{\refset(\rho)}$ if
%    $\Tr_{\rM_\rF}[\Psi]=\sigma$, with $\sigma\in\refset
%(\rho)$.
%\item $\Psi\in\st(\rL_\rF\rM_\rF)$ is in $\refset(D_{\rho})$ if there
%  exists  states $\Omega,\Sigma\in\st{(\rL_\rF\rM_\rF)}$, with
%  $\Omega\in D_\rho$, and $\Omega=\Sigma+\Psi$.
%\end{enumerate}
%Now we show that the two conditions above coincide:
%\noindent
%1$\Rightarrow$2: Since $\sigma\in\refset(\rho)$, it is
%$\rho=\sigma+\tau$ for some $\tau\in\st(\rL_\rF)$. Consider now the
%state $\Sigma=\tau\otimes \omega$, with $\omega$ and arbitrary state
%of $\st{(\rM_\rF)}$. It is now trivial to verify that
%$\Omega=\Sigma+\Psi$ is a state of $D_\rho$.
%
%\noindent
%2$\Rightarrow$1: By definition $\Omega=\Sigma+\Psi$ is in $D_\rho$,
%namely $\Tr_{\rM_\rF}[\Omega]=\rho$. Accordingly, one has both
%$\Tr_{\rM_\rF}[\Sigma]\in \refset(\rho)$ and
%$\Tr_{\rM_\rF}[\Psi]\in \refset(\rho)$.
%
%\end{proof}

We can now prove the following proposition and the subsequent corollary
providing a simple reliability criterion for fermionic compression.

\begin{proposition}
  Given a state $\rho\in\st_{1}(\rL_\rF)$ and a channel
  $\tC\in\tr_1(\rL_\rF\to\rL_\rF)$, $\forall\varepsilon>0$ there
  exists $\delta>0$ such that if $F(\rho,\tC)>1-\delta$ then
  $\sum_i\left\lVert
    [(\tC-\tI)\boxtimes\tI](\Sigma_i)\right\rVert_{1}\leq\varepsilon$
  for every $\{\Sigma_i\}$ 
  such that $\sum_i\Sigma_i\in D_\rho$.
\end{proposition}
\begin{proof}
  Firstly we observe that, given a set of states $\{\Sigma_i\}$ such that 
  $\sum_i\Sigma_i\in D_\rho$, considering any purification $\Psi_\rho\in\purst(\rL_\rF\rK_\rF\rN_\rF)$ of 
  $\Sigma\coloneqq\sum_i\Sigma_i$, one can  find a POVM $\{b_i\}\in\eff(\rN_\rF)$ such that
  $\Sigma_i=\Tr_{\rN_\rF}[(I_{\rL_\rF\rK_\rF}\boxtimes b_i)\Psi_\rho]$.
  As a consequence, we have
  \begin{align*}
  \sum_i&\norma{(\tC-\tI)\Sigma_i}_1\\
  &=\sum_i\norma{\Tr_{\rN_\rF}[\{(\tC-\tI)\boxtimes\tI_{\rN_\rF}\}(\Psi_\rho)(I_{\rL_\rF\rK_\rF}\boxtimes b_i)]}_1\\
  &\leq\norma{\{(\tC-\tI)\boxtimes\tI_{\rN_\rF}\}(\Psi_\rho)}_1,
  \end{align*}
  where the last inequality follows from the equivalent definition of the 1-norm for $X\in\st_\mathbb R(\rL_\rF)$
  \begin{align*}
  \|X\|_1=\max_{b\in\eff(\rL_\rF)}\Tr[Xb],
  \end{align*}
  and from the fact that, for $\{a_i\}\subseteq\eff(\rL_\rF\rK_\rF)$ such that 
  \begin{align*}
  \norma{\Tr_{\rN_\rF}[&\{(\tC-\tI)\boxtimes\tI_{\rN_\rF}\}(\Psi_\rho)(I_{\rL_\rF\rK_\rF}\boxtimes b_i)]}_1\\
  &=\Tr_{\rN_\rF}[\{(\tC-\tI)\boxtimes\tI_{\rN_\rF}\}(\Psi_\rho)(a_i\boxtimes b_i)],
  \end{align*}
  one can write
  \begin{align*}
    \sum_i\norma{(&\tC-\tI)\Sigma_i}_1\\
    &=\sum_i\Tr_{\rN_\rF}[\{(\tC-\tI)\boxtimes\tI_{\rN_\rF}\}(\Psi_\rho)(a_i\boxtimes b_i)]\\
    &=\Tr_{\rN_\rF}[\{(\tC-\tI)\boxtimes\tI_{\rN_\rF}\}(\Psi_\rho)A],
  \end{align*}
  where $A\coloneqq \sum_i(a_i\boxtimes b_i)$.
  Now, by the Fuchs-van der Graaf inequalities, if $F(\rho,\tC)\geq1-\delta$, then
  \begin{align*}
  \norma{\{(\tC-\tI)\boxtimes\tI_{\rN_\rF}\}(\Psi_\rho)}_1\leq2\sqrt{1-F(\rho,\tC)}\leq2\sqrt{\delta}.
  \end{align*}
  The thesis is then obtained just taking $\delta\leq\varepsilon^2/4$.
  \end{proof}

\begin{corollary}[Reliable compression scheme] Given a state $\rho\in\st(\rL_\rF)$, a
  compression scheme $(\tE^{N},\tD^{N})$ is $\epsilon$-reliable if one has
  $F(\rho^{\boxtimes N},\tC_N)>1-\delta$, where  $\delta=\epsilon^2/4$, and $\tC_N:=\tD^{N}\tE^{N}$.
\end{corollary}

\subsection{Fermionic typical subspace} 

At the basis of the quantum source coding theorem lies the notion of typical
subspace, that in turn generalizes to the quantum case that of typical
sequences and typical sets of classical information. We now introduce
the notion of typical subspace also for fermionic systems and use it to show
that, like in quantum theory, the von Neumann entropy of a fermionic state is the 
rate that separates the region of rates for which a reliable compression scheme exists
from that of unachievble rates. In order to do this we
have to verify that the compression map given in terms of the
projection on the typical subspace represents an admissible fermionic
map. 

We start by defining the notion of logarithm of a fermionic state
\begin{definition}\label{def:logarithm} 
Let $\rho$ be a fermionic state. We define its logarithm as
\begin{align}
\log_{2}\rho=J^{-1}[\log_2 J(\rho)].
\end{align}
\end{definition}
Then we define the von Neumann entropy of a fermionic state via its Jordan-Wigner
representative.
\begin{definition}\label{def:fermionic_entropy}
Given a fermionic state $\rho$, its von-Neumann entropy is defined as 
\begin{align}
S_f(\rho):=S(J(\rho))=-\Tr(J(\rho)\log_{2}J(\rho)).
\end{align}
\end{definition}
These definitions are independent of the particular Jordan-Wigner transform corresponding
to a given ordering of the modes (see Appendix~\ref{app:well-defined-operations}).

 When we use the orthonormal decomposition for $J(\rho)=\sum_{x_i} p_i\proj{x_i}$, this reduces to the
Shannon entropy of the classical random variable $X$ that takes values
in $\Rng(X)=\{x_1,x_2,\ldots x_n\}$, called range of $X$, with
probability distribution $(p_{1},p_{2},\ldots, p_{n})$:
$S_f(\rho)=H(X)=-\sum_{i}p_{i}\log_{2}p_{i}$. We remind that $N$
i.i.d.~copies of the state $\rho$ are represented as
$ J(\rho^{\boxtimes N})=J(\rho)^{\otimes N}
=\sum_{x_\vi\in\Rng(X)^N}p_\vi\proj{x_\vi}, $. With
$\Typ_{N,\varepsilon}(\rho)$ we will denote the typical set of the
random variable $X$.
\begin{definition}[Typical subspace] Let $\rho\in\st(\rL_\rF)$ with orthonormal decomposition
  $J(\rho)=\sum_{x_i\in\Rng(X)} p_i\proj{x_i}$. The $\varepsilon$-{\em
    typical} subspace $\QTyp_{N,\varepsilon}(\rho)$ of
  $\cH^{\otimes N}_{L}$ is defined as
\begin{equation}
\QTyp_{N,\varepsilon}(\rho):=\Span\{\ket{x_\vi}\ |\  x_\vi\in\sT_{N,\varepsilon}(X)\},
\end{equation}
where $\ket{x_\vi}:=\ket{x_{i_1}}\ket{x_{i_2}}\ldots\ket{x_{i_N}}$,
and $X$ is the random variable with $\Rng(X)=\{x_i\}$ and
$\mathbb P_X(x_i):=p_i$.
\end{definition}

It is an immediate consequence of the definition of typical subspace
that
\begin{equation*}
  \QTyp_{N,\varepsilon}(\rho):=\Span\left\{\ket{x_\vi}\ |\ \left|\frac{1}{N}\log_2 \frac{1}{\mathbb P_{X^N}(x_\vi)}
      -S_f(\rho)\right|\leq\varepsilon\right\}.
\end{equation*}
We will denote the projector on the typical subspace as
\begin{equation}\label{eq:typ_space_proj}
\begin{aligned}
P_{N,\varepsilon}(\rho):=&\sum_{x_\vi\in\Typ_{N,\varepsilon}(X)}\proj{x_\vi}\\
=&
\sum_{x_\vi\in\sT_{N,\varepsilon}(X)}\proj{x_{i_1}}\otimes\cdots\otimes \proj{x_{i_N}},
\end{aligned}
\end{equation}
and we have that
$\dim(\QTyp_{N,\varepsilon}(\rho))=\Tr[P_{N,\varepsilon}(\rho)]=|\Typ_{N,\varepsilon}(X)|$.
%Notice that the definition of $\QTyp_{N,\varepsilon}\equiv\QTyp_{N,\varepsilon}(\rho)$ relies on $\rho$, however, since we will keep $\rho$ fixed, we will not write such dependence explicitly, in an analogous way as for the classical
%set $\Typ_{N,\epsilon}\equiv\Typ_{N,\epsilon}(X)$ which relies on the random variable $X$.

Notice that some of the superpositions of vectors in the typical subspace might not be 
legitimate fermionic pure states, as their parity might be different. However, up to now,
we only defined the typical subspace as a mathematical tool, and it does not need a consistent
physical interpretation. We will come back to this point later (see Lemma~\ref{lem:fermionic-projector}), when we will discuss the 
physical meaning of the projection $P_{N,\varepsilon}(\rho)$.
%above defined projector might not be a legitimate fermionic operator, as 
%the parity of the vectors $\ket{x_\vi}$ is not fixed. However, one can easily realise that
%\begin{align}
%P_{N,\varepsilon}(\rho)=P^{(E)}_{N,\varepsilon}(\rho)+P^{(O)}_{N,\varepsilon}(\rho),
%\end{align}
%where $P^{(X)}_{N,\varepsilon}(\rho)$ projects on the span of those vectors $\ket{x_\vi}$ such that 
%$x_\vi\in T_{N,\varepsilon}(\rho)$ and $\ket{x_\vi}$ has parity $X$. Now, as $\rho$ is block-diagonal 
%with respect to the even and odd decomposition, we can conclude that
%\begin{align*}
%P_{N,\varepsilon}(\rho)\rho P_{N,\varepsilon}(\rho)=P^{(E)}_{N,\varepsilon}(\rho)\rho P^{(E)}_{N,\varepsilon}(\rho)+P^{(O)}_{N,\varepsilon}(\rho)\rho P^{(O)}_{N,\varepsilon}(\rho),
%\end{align*}
%which is a legitimate fermionic map. 
Now, it is immediate to see that
\begin{align}
&\Tr[P_{N,\varepsilon}(\rho)J(\rho)^{\otimes N}]\nonumber\\
&\qquad=\sum_{x_\vi\in\sT_{N,\varepsilon}(\rX)}\mathbb P_{X^N}(
x_\vi)=\mathbb P_{\rX^N}[x_\vi\in\Typ_{N,\varepsilon}(X)]. 
\label{Qxtyp}
\end{align}
As in quantum theory, also the fermionic typical subspace has the
following features:
\begin{proposition}[Typical subspace]\label{thm:typical-subspace} Let $\rho\in\st(\rL_\rF)$. The following statements hold:
\begin{enumerate}[leftmargin=*]
\item For every $\varepsilon>0$ and $\delta>0$ there exists $N_0$ such that for every $N\geq N_0$
\begin{equation}
\Tr[P_{N,\varepsilon}(\rho)J(\rho)^{\otimes N}]\geq 1-\delta.
\end{equation}
\item \label{it:qcard}For every $\epsilon>0$ and $\delta>0$ there
  exists $N_0$ such that for every $N\geq N_0$ the dimension of the
  typical subspace $\QTyp_{N,\varepsilon}(\rho)$ is bounded as
\begin{equation}
(1-\delta)2^{N(S_f(\rho)-\varepsilon)}\leq\dim(\QTyp_{N,\varepsilon}(\rho))\leq 2^{N(S_f(\rho)+\varepsilon)}
\end{equation}
\item \label{it:nontriv}For given $N$, let $S_N$ denote an arbitrary
  orthogonal projection on a subspace of $\cF_{L}^{\otimes N}$ with
  dimension $\Tr(S_N)< 2^{NR}$, with $R<S_f(\rho)$ fixed. Then for every
  $\delta>0$ there exists $N_0$ such that for every $N\geq N_0$ and
  every choice of $S_N$
\begin{equation}
\Tr[S_NJ(\rho)^{\otimes N}]\leq\delta.
\label{eq:condsn}
\end{equation}
\end{enumerate}
\end{proposition}

The proof of the above properties is exactly the same as the one of
quantum theory (see for instance \cite{QInielsenchuang}). However,
order to exploit the same scheme proposed by Schumacher for the
quantum case, one has to check that the encoding and decoding channels
given in the constructive part of the proof
%terms of the projector $P_{N,\varepsilon}(\rho)$ on the
%typical subspace 
are admissible fermionic maps. In particular, the
encoding channel makes use of the projector $P_{N,\varepsilon}(\rho)$
as a Kraus operator, therefore, we have to show that it is a legitimate
Kraus for a fermionic map. This is proved in the following lemma based
on characterization of fermionic transformations of
Proposition~\ref{prop:fermionickraus}.

\begin{lemma}\label{lem:fermionic-projector}
Let $\rho$ be a fermionic state.  The projector $P_{N,\varepsilon}(\rho)$ of eq \ref{eq:typ_space_proj} is the Kraus
  operator of an admissible fermionic transformation.
\end{lemma}
\begin{proof}
  By proposition \ref{prop:fermionickraus} the projector on the
  typical subspace $P_{N,\varepsilon}(\rho)$ is a legitimate fermionic Kraus
  if it is the sum of products of either an even or an odd number of
  fermionic fields. Let us consider the single projection
  $\ket{x_{\vi}}\bra{x_{\vi}}$. This is given by the tensor product
  $\ket{x_{i_{1}}}\bra{x_{i_1}}\otimes\dots\otimes\ket{x_{i_{N}}}\bra{x_{i_N}}$, where each
  $\ket{x_{i_{k}}}$ is an eigenvector of the density matrix $J(\rho)$
  representing the fermionic state $\rho$, and, as such, it has a definite parity. 
  Thus, each factor in the above expression of $\ket{x_{\vi}}\bra{x_{\vi}}$ is the Jordan-Wigner 
  representative of an even polynomial, and also the projection $\ket{x_{\vi}}\bra{x_{\vi}}$ 
  is thus the representative of an even polynomial for every $\vi$, which is given, in detail, 
  by the product $J^{-1}(\ket{x_{\vi}}\bra{x_{\vi}})=\prod_{j=1}^N J^{-1}(\ket{x_{i_j}}\bra{x_{i_j}})$. 
  Now, 
%  since 
%  $J^{-1}(\ket{x_{\vi}}\bra{x_{\vi}})$ is a product of an even number 
%  of field operators, 
	by Proposition~\ref{prop:fermionickraus}, 
%it is a 
%  fermionic effect.
%  can be represented as
%  $\ket{x_{i_{k}}}\equiv Q_{i_{k}}(\phi_{\ell})\ket{\Omega}$ where
%  $Q_{i_{k}}$ is a polynomial with definite parity with respect to
%  fermionic fields. This implies that
%  $\ket{x_{\vi}}=Q_{i_{1}}\otimes\dots\otimes Q_{i_{N}}
%  \ket{\Omega}^{\otimes N}$
%  is also given by a polynomial with definite parity applied to the
%  tensor product of vacuum states. 
%  Hence, 
%  each single projector
%  $\ket{x_{\vi}}\bra{x_{\vi}}$ is the Jordan-Wigner representative of a sum of products of an
%  even number of fermionic fields and 
  $P_{N,\varepsilon}(\rho)$ is the Jordan-Wigner representative of a
  legitimate fermionic Kraus operator.
\end{proof}

\subsection{Fermionic source coding theorem} 

We can now prove the source coding theorem for
fermionic information theory.

\begin{theorem}[Fermionic source coding]\label{thm:source-coding} 
Let $\rho\in\st_{1}(\rL_\rF)$ be a
  state of system $\rL_\rF$. Then for every $\delta>0$ and $R>S_f(\rho)$
  there exists $N_0$ such that for every $N\geq N_0$ one has a
  compression scheme $\{\tE_N,\tD_N\}$ with rate $R$, and
  $F(\rho^{\boxtimes N},\tD_N\tE_N)\geq1-\delta$. Conversely,
  for every $R<S_f(\rho)$ there is $\delta\geq0$ such that for every
  compression scheme $\{\tE_N,\tD_N\}$ with rate $R$ one has
  $F(\rho^{\boxtimes N},\tD_N\tE_N)\leq\delta$.
\end{theorem}
The proof follows exactly the lines of the original proof for standard quantum compression, that can be found 
e.g.~in Ref.~\cite{QInielsenchuang}. As the direct proof is constructive, we only need to take care of the
legitimacy of the compression protocol as a fermionic map. To this end, we recapitulate the construction here.

%  Let $R>S(\rho)$. Take $\varepsilon>0$ such that
%  $R\geq S(\rho)+\varepsilon$. Then, by the typical subspace theorem
%  one has that for every $\delta>0$ there exists $N_0$ such that for
%  every $N\geq N_0$. Now, consider the following compression scheme.
%
  \begin{enumerate}[leftmargin=*]
  \item Encoding: Perform the measurement
    $\{P_{N,\varepsilon}(\rho), I-P_{N,\varepsilon}(\rho)\}$. If the
    outcome corresponding to $P_{N,\varepsilon}(\rho)$ occurs, then
    leave the state unchanged. Otherwise, if the outcome corresponding
    to $I-P_{N,\varepsilon}(\rho)$ occurs, replace the state by a
    standard state $\proj{S}$, with
    $\ket{S}\in\QTyp_{N,\varepsilon}(\rho)$.  Such an map is
    described by the channel $\tM_N:\rL_\rF^{\boxtimes N}\to\rL^{\boxtimes N}_\rF$ given by
\begin{align*}
  &J(\tM_N)(\sigma)\coloneqq\\
  &\quad P_{N,\varepsilon}(\rho)\sigma P_{N,\varepsilon}(\rho)+\Tr[(I-P_{N,\varepsilon}(\rho))\sigma]\ket{S}\bra{S}
\end{align*}
Notice that this is a well defined transformation since by
Lemma~\ref{lem:fermionic-projector} the projector on the typical
subspace is a legitimate fermionic Kraus operator. The second term is a
measure and prepare channel, which is also a legitimate fermionic
transformation. Then consider a system $\rM_\rF$ made of
$M:=N\ceil{R}$ LFMs and the (partial) isometric embedding
$V:\QTyp_{N,\varepsilon}(\rho)\rightarrow\cH_{N\ceil{R}}$ such that
$V^{\dagger}V=I_{\QTyp_{N,\varepsilon}(\rho)}$. Since the first stage 
of the protocol never produces states in the complement of $\QTyp_{N,\varepsilon}(\rho)$,
we can complete the map $V\cdot V^\dag$ to a fermionic channel $\tV_N$. The encoding is then given by the composite map $\tE_N:=\tV_N\tM_N$.

\item Decoding: For the decoding channel, we simply choose the co-isometry $V^{\dagger}$, which inverts $V$
on $\QTyp_{N,\varepsilon}(\rho)$.
\end{enumerate}
As for the converse statement, the proof for quantum compression is based on item~\ref{it:nontriv}, which 
we proved for fermionic theory as well. Thus, the quantum proof applies to the fermionic case.

\section{Discussion}

We have studied information compression for fermionic systems, showing
the fermionic counterpart of the quantum source coding theorem. In
spite of parity superselection rule and the non locality of the Jordan-Wigner
representation of fermionic operators, the von Neumann entropy of
fermionic states can still be interpreted as their information content,
providing the minimal rate for which a reliable compression is
achievable.

The novelty in this paper is the analysis of compression in the
absence of local tomography. Here, the properties of a map, and in the
specific case of study of the compression map, cannot be accessed
locally. This poses stronger constraints on the set of reliable
compression maps. 

Despite the significant differences between fermionic and quantum
information~\cite{DAriano2014}, the source coding theorem holds also
for fermions. We can now wonder which are the minimal features of a
theory that lie behind the coding theorem. As we learn from classical,
quantum, and now also fermionic information theory, the task of
information compression is intimately related to the notion of
entropy. However, it is known that information theories beyond quantum
exhibit inequivalent notions of
entropy~\cite{KIMURA2010175,Barnum_2010,Short_2010}. This is the main
issue one has to face in order to introduce the notion of information
content in the general case. On one side one has to provide a definition of
information content including a broad class of probabilistic
theories. On the other side one can compare such a notion with the
different notions of entropy, identifying the one that plays the same
role of Shannon entropy in the compression task.

\acknowledgments A.T. acknowledges financial support from the Elvia
and Federico Faggin Foundation through the Silicon Valley Community
Foundation, Grant No. 2020-214365. This work was supported by MIUR
Dipartimenti di Eccellenza 2018-2022 project F11I18000680001.

\bibliographystyle{unsrt}
\bibliography{biblio}

\appendix
\section{Fermionic States}\label{app:appstates}
In a $L$-LFM system, fermionic states in $\st(\rL_{\rF})$ are represented by density matrices on the antisymmetric Fock space $\cF_{L}$
satisfying the parity superselection rule. As such they can be written as combinations of products of field
operators. Indeed, a fermionic state $\rho$ can be split in its even and odd part as follows
\[
\begin{aligned}
\rho=\sum_{e}E_{e}\ket{\Omega}\bra{\Omega}E^{\dagger}_{e}+\sum_{o}O_{o}\ket{\Omega}\bra{\Omega}O^{\dagger}_{o},
\end{aligned}
\]
where $E_{e}$ and $O_{o}$ are linear combinations of products of even and odd number of
field operators respectively. By recalling that $\ket{\Omega}\bra{\Omega}=\prod_{i=1}^{L}\varphi_{i}\varphi_{i}^{\dagger}$ one can easily realize
that $\rho$ can be written as combination of products of even number of field operators. Moreover, by using the CAR, the generic state can be written as follows
\[
\rho=\sum_{\underline{s},\underline{t}}
	\rho_{\underline{s}\underline{t}}\prod_{i=1}^{L}\varphi_{i}^{\dagger s_{i}}\varphi_{i}\varphi_{i}^{\dagger}\varphi_{i}^{t_{i}},
\] 
where $\underline{s},\underline{t}\in\{0,1\}^L$ and $\rho_{\underline{s}\underline{t}}\in \mathbb{C}$.

\section{Technical Lemmas}\label{app:app1}

Here we show two lemmas that are used in the proof of Proposition~\ref{def:purification} in the main text.

As a preliminary notion we define quantum states with definite parity. Let $\cH_L$ be an Hilbert space of $L$-qubits and let $\st{(\cH_L)}$ be the corresponding set of states. 
The vectors of the computational basis 
\begin{align}
\ket{s_1,s_2,\ldots, s_L},\quad s_i=\{0,1\},\quad i=1,\ldots L,
\end{align}
can be divided into even $p=0$ and odd $p=1$ vectors according to their parity $p:=\oplus_{i=1}^L s_i$. 
Denoting by $\cH_L^0$ and $\cH_L^1$, with $\cH_L=\cH_L^0\oplus\cH_L^1$, the spaces generated by even and odd vectors respectively, one says that a state $\rho\in\st{(\cH_L)}$ has definite parity if it is of the form $\rho=\rho_0+\rho_1$, with $\rho_0$ and $\rho_1$ having support on $\cH_L^0$ and $\cH_L^1$ respectively. As a special case, a pure state of definite parity $p$ must have support only on $\cH_L^p$. We can now prove the following lemma.

 \begin{lemma}\label{lem:techlemma0}
Consider a quantum state $\rho\in\st{(\cH_L)}$ and two purifications $\Psi,\Phi\in\st{(\cH_L \cH_M)}$ with definite parity. Then it is alway possible to find a unitary channel $\tU$ that maps states of definite parity into states of definite parity and such that $(\tI\otimes\tU)(\Psi)=\Phi$.
\end{lemma} 
\begin{proof}

%Without loss of generality let us assume that $\Psi$ and $\Phi$ have the same parity. Indeed, if 
% $\Psi\in\cH_{LM}^e$ and $\Phi_{LM}^o$, one can take $\bigotimes_{i=1^{j-1}\sigma^x_j\bigotimes_{i=1}^{j-1}
% 
% $(I\otimes\sigma_j^x$

Let $\ket{\Psi}\in\cH_{LM}^p$ and  $\ket{\Phi}\in\cH_{LM}^q$, for $p,q\in\{0,1\}$. Since the two states are purification of the same state $\rho\in\st{(\cH_L)}$ their Schmidt decomposition can always be taken as follows
\begin{align*}
\ket{\Psi}=\sum_{i} \lambda_i\ket{i}\ket{\Psi_i},\qquad\ket{\Phi}=\sum_{i} \lambda_i\ket{i}\ket{\Phi_i},
\end{align*}
where $\{\ket{i}\}\in\cH_L$ is the same orthonormal set for the two states, while $\{\ket{\Psi_i}\},\{\ket{\Phi_i}\}\in\cH_M$ are two generally different orthonormal sets. Notice that, since $\Psi$ and $\Phi$ are pure states of definte parity, any element in the above orthonormal sets must be a vector of definite parity. Within the set $\{\ket{i}\}=\{\{\ket{i_0}\},\{\ket{i_1}\}\}$ one can separate  even $\{\ket{i_0}\}$ and odd $\{\ket{i_0}\}$ parity vectors, and then write  $\Psi$ and $\Phi$ (respectively of parity $p$ and $q$) as
\begin{align*}
&\ket{\Psi}=\sum_{i_0} \lambda_{i_0}\ket{i_0}\ket{\Psi^p_{i_0}}+\sum_{i_1}\lambda_{i_1}\ket{i_1}\ket{\Psi^{\bar{p}}_{i_1}},\\
&\ket{\Phi}=\sum_{i_0} \lambda_{i_0}\ket{i_0}\ket{\Phi^q_{i_0}}+\sum_{i_1}\lambda_{i_1}\ket{i_1}\ket{\Phi^{\bar{q}}_{i_1}},
\end{align*}
where $\bar{r}=r\oplus 1$, and in the orthonormal sets $\{\ket{\Psi^p_{i_0}},\ket{\Psi^{\bar{p}}_{i_1}}\}$ and $\{\ket{\Phi^q_{i_0}},\ket{\Phi^{\bar{q}}_{i_1}}\}$ we separated vectors according to their parity. We can now complete the above two sets to orthonormal bases in such a way that all vectors in both bases have definite parity. Let us take for example the basis 
$\{\ket{\Psi^p_{i_0}},\ket{\Psi^{\bar{p}}_{i_1}}\},| \Psi_k^{r(k)}\rangle\}$ and $\{\ket{\Phi^1_{i_0}},\ket{\Phi^{\bar{q}}_{i_1}}, |\Phi_k^{t(k)}\rangle \}$ with $r(k),t(k)\in\{0,1\}$. It is now straightforward to see that the unitary map $\tU$ having Kraus operator
\begin{align*}
U=
\sum_{i_0}\ket{\Psi^p_{i_0}}\bra{\Phi^q_{i_0}}+\sum_{i_1}\ket{\Psi^{\bar{p}}_{i_1}}\bra{\Phi^{\bar{q}}_{i_1}}+\sum_k |\Psi^{r(k)}_k\rangle\langle\Phi^{t(k)}_k|
\end{align*} 
is such that $(I\otimes U)\ket{\Psi}=\ket{\Phi}$. Moreover $\tU$ maps states of definite parity into states of definite parity.
\end{proof}

\begin{lemma}\label{lem:techlemma}
Let $\rN_{\rF}:=\rL_{\rF}\rK_{\rF}$ and $\tC\in\tr(\rN_{\rF}\rightarrow \rN_{\rF})$ be a single Kraus transformation with
Kraus $C$ having Jordan-Wigner rapresentative $J(C)=U\otimes I_{\rK_{F}}$, $U$ acting on the first $L$ qubits. Then $\tC$ is local on the first $L$ modes.
\end{lemma}
\begin{proof}
Due to Proposition~\ref{prop:fermionickraus}, the Kraus operator of $\tC$
can be written as $C=\sum_{i}C_{i}$, where either each $C_i$ is a product of an even number of field
operators, or each $C_i$ is a product of an odd one. 
The set $\{C_{i}\}$ can be taken to be linearly independent without 
loss of generality. 
Let us assume by contradiction that $\tC$ is not local on the first $L$ modes. Therefore, since a set of 
independent operators generating the algebra of the $j$-th mode is $\{\varphi_{j},\varphi_{j}^\dag,\varphi_{j}^\dag\varphi_{j},\varphi_{j}^\dag\varphi_{j}+\varphi_{j}\varphi_{j}^\dag\}$, 
there exists at least  
one product $C_i$ that contains one of the factors $\varphi_{j}$, $\varphi_{j}^{\dagger}$, or $\varphi_{j}\varphi_{j}^{\dagger}$,
for some mode $j$ of the system $\rK_\rF$. 
%Given the subset of the $K$ modes containing one of those terms, 
Let
$j(i)$ be the mode with largest label in the chosen ordering of the $N=L+K$ modes, such that the corresponding 
factor in the product $C_i$ is not the identity (i.e. $\varphi_{j}^\dag\varphi_{j}+\varphi_{j}\varphi_{j}^\dag$). 
Accordingly, one has that the Jordan-Wigner representative of $C_i$ is of the form
\begin{align}\nonumber
&J(C_i)=K\otimes O_{j(i)} \otimes \left(\bigotimes_{l=j(i)+1}^{N} I_l\right),
\end{align}
where $K$ is an operator on the first $1,\ldots,{j(i)-1}$ qubits, and $O_{j(i)}$ is one of the factors 
$\sigma_{j(i)}^{+},\sigma_{j(i)}^{-}, \sigma_{j(i)}^{+}\sigma_{j(i)}^{-}$ on the $j$-th qubit.
This  contradicts the hypothesis on the form of $J(C)$.
\end{proof}

\section{Jordan-Wigner independence}\label{app:well-defined-operations}
In this appendix we show the consistency of definitions \ref{def:squareroot}, \ref{def:ent_fid}, 
\ref{def:logarithm} and \ref{def:fermionic_entropy} given in text. In particular,
we check that they are independent of the particular choice of the order of the fermionic modes, which defines
the Jordan-Wigner transform. We remember that all Jordan-Wigner representations are unitarily equivalent.

\begin{lemma}
Let $\rho$ be a fermionic state. The square root and the logarithm of $\rho$ are well defined.
\end{lemma}
\begin{proof}
Once we have fixed the ordering of the modes, the square root of a fermionic state $\rho$ is defined via its Jordan-Wigner representative as follows 
\[
\rho^{\frac{1}{2}}:= J^{-1}[J(\rho)^{\frac{1}{2}}]
\]
If $\tilde{J}$ is the Jordan-Wigner isomorphism associated to a different ordering,
then consider $X:= \tilde{J}^{-1}[\tilde{J}(\rho)^{\frac{1}{2}}]$. We can now prove that  
$X=\rho^\frac{1}{2}$ and then independence of the square root from  the ordering. Indeed, one has 
\[
\tilde{J}(X)^{2}=\tilde{J}(\rho)=UJ(\rho)U^{\dagger},
\]
with $U$ unitary. It follows that
\[
J(\rho)=U^{\dagger}\tilde{J}(X)UU^{\dagger}\tilde{J}(X)U=J(X)^{2}\implies J(X)=J(\rho)^{\frac{1}{2}}.
\]
Since $J$ is an isomorphism, by taking $J^{-1}$ we finally get
\[
X=J^{-1}[J(\rho^{\frac{1}{2}})]=\rho^{\frac{1}{2}}.
\] 

Analogously, the logarithm of a fermionic state is defined thorugh its Jordan-Wigner representative
\[
\log_{2}(\rho):=J^{-1}[\log_{2}(J(\rho))]
\]
Again, let $\tilde{J}$ be the Jordan-Wigner isomorphism corresponding to a different ordering,
and let $X=\tilde{J}^{-1}[\log_{2}(\tilde{J}(\rho))]$. Firstly we notice that
\[
\log_{2}[\tilde{J}(\rho)]=\log_{2}[UJ(\rho)U^{\dagger}]=U\log_{2}[J(\rho)]U^{\dagger}
\]
since $U$ is unitary (we remind that the logarithm of a positive operator is defined via its spectral decomposition, and a unitary map preserves the spectrum). Therefore, we find
\[
J(X)=\log_{2}[J(\rho)]\implies X=J^{-1}[\log_{2}(J(\rho))]=\log_{2}(\rho). 
\]
that concludes the proof.
\end{proof}

Based on the above lemma we have the following proposition.
\begin{proposition}
Let $\rho$ and $\sigma$ be two fermionic states. The Uhlmann fidelity $F(\rho,\sigma)$ and 
the von Neumann entropy $S_{f}(\rho)$ of definitions \ref{def:ent_fid} and \ref{def:fermionic_entropy} are well defined. 
\end{proposition}
\begin{proof} These two quantities are given by a trace of two well defined  operators, as proved in the previous
lemma. Moreover, since a reordering of the modes corresponds to a unitarily change of basis, the trace is 
Jordan-Wigner independent, and so are $F(\rho,\sigma)$ and $S_{f}(\rho)$.

\end{proof}

% \section{}

% In what follows, $\norma{T}_{p}$ will denote the Shatten p-norm of the
% operator $T\in \Lin(\cH_{A}\rightarrow\cH_{B})$ which is defined as:
% \begin{equation*}
% \norma{T}_{p}={\rm Tr}[(T^{\dagger}T)^{p/2}]^{1/p},
% \end{equation*}
% and the $\infty$-shatten norm corresponds to the sup norm of $T$ (see for instance \cite{bhatia97})
% \begin{equation*}
% \norma{T}_{\infty}=\sup \{ T\phi : \phi \in \cH, \norma{\phi}\leq 1\}.
% \end{equation*}
% A relevant inequality is~\cite{QIwatrous}
% \begin{equation}
% \norma{STR}_{p} \leq \norma{S}_{\infty} \norma{T}_{p} \norma{R}_{\infty},
% \label{eq:normprop}
% \end{equation}
% where $S$, $T$, $R$ are linear operators defined on suitable Hilbert spaces.
% The following theorem on the equivalence between the trace-norm (Shatten 1-norm) and the Uhlmann's distances is also used
% \cite{QInielsenchuang}

\end{document}